\documentclass[conference]{IEEEtran}
\usepackage[utf8]{inputenc}
\usepackage{amsmath}
\usepackage{amsfonts}
\usepackage{amsthm}
\usepackage{tikz}

\IEEEoverridecommandlockouts

\newtheorem{thm}{Theorem}

\newtheorem{cor}{Corollary}

\title{Fisher Information and Mutual Information Constraints}
\author{Leighton Pate Barnes and Ayfer \"Ozg\"ur\\
Stanford University, Stanford, CA 94305\\
 Email: \{lpb, aozgur\}@stanford.edu
\thanks{This work was supported in part by NSF award CCF-1704624 and by a Google Faculty Research Award. Thanks to an anonymous reviewer for catching a bug in the proof of Corollary 1.}}

\begin{document}

\maketitle

\begin{abstract}
 We consider the processing of statistical samples $X\sim P_\theta$ by a channel $p(y|x)$, and characterize how the statistical information from the samples for estimating the parameter $\theta\in\mathbb{R}^d$ can scale with the mutual information or capacity of the channel. We show that if the statistical model has a sub-Gaussian score function, then the trace of the Fisher information matrix for estimating $\theta$ from $Y$ can scale at most linearly with the mutual information between $X$ and $Y$. We apply this result to obtain minimax lower bounds in distributed statistical estimation problems, and obtain a tight preconstant for Gaussian mean estimation. We then show how our Fisher information bound can also imply mutual information or Jensen-Shannon divergence based distributed strong data processing inequalities.
\end{abstract}

\section{Introduction}

In this work, we consider the processing of statistical samples $X\sim P_\theta$ by a channel $p(y|x)$, and try to understand how the statistical information from the samples for estimating the parameter $\theta$ can scale with the mutual information or capacity of the channel. In particular, we begin by looking at the Fisher information for estimating $\theta$ from the processed data $Y$. Fisher information describes the curvature of the statistical model as one moves around the parameter space, and immediately implies lower bounds for the error in estimating $\theta$ from the statistical samples via the well-known Cram\'er-Rao lower bound \cite{cramer,rao} for unbiased estimators. It can also imply lower bounds for arbitrarily biased estimators in an asymptotic sense \cite{van2000asymptotic}, or in a Bayesian setting via the van Trees inequality \cite{gill}.

Fisher information satisfies a data processing inequality in the sense that it must decrease during processing \cite{zamir}. In our main result, we develop a strong data processing inequality that more precisely quantifies the maximum possible Fisher information from $Y$ -- in terms of the mutual information between $X$ and $Y$ -- after processing the data. This is a generalization of recent results on Fisher information due to the authors in both the communication constrained \cite{barnes} and privacy constrained \cite{ldp_fisher} settings. Other works such as \cite{acharyaetal} have also considered statistical inference under general channels and obtain some bounds for communication and privacy constrained channels in particular. In concurrent work by the same authors \cite{acharya2}, they obtain bounds for mutual information constrained channels that are similar to those in the present paper. Their work differs in that they only treat discrete output spaces $\mathcal{Y}$, they do not obtain lower bounds for fully interactive communication protocols, and they use a slightly different sub-Gaussianity assumption. One application where general mutual information constrained channels are needed is when communication of statistical samples is done over an analog channel such as an additive white Gaussian noise channel or Gaussian multiple access channel, such as in recent works \cite{CZ1} and \cite{CZ2} which seek to jointly study the communication and estimation problem.

In our main result, we show that if the score of the statistical model is sub-Gaussian -- an assumption that is also used in previous works \cite{barnes,ldp_fisher} -- then the trace of the Fisher information matrix for estimating $\theta$ from $Y$ can scale at most linearly with the mutual information between $X$ and $Y$. We also show by the example of the Gaussian location model that the preconstant we obtain is optimal and cannot be improved. We then apply this result in a distributed setting where there are $n$ nodes each with a sample $X_i$ taken independently from $P_\theta$, and where the nodes can communicate in multiple rounds of communication via a public blackboard. We show how Fisher information from the total blackboard transcript can be similarly bounded, and develop minimax lower bounds in this distributed estimation setting. Finally, in the last section, we show how our Fisher information upper bound can also imply mutual information or Jensen-Shannon divergence based distributed strong data processing inequalities similar to those from \cite{braverman}.

\section{Preliminaries}

Suppose that $\{P_\theta\}_{\theta\in\Theta}$ is a family of probability distributions parametrized by $\theta\in\Theta\subseteq\mathbb{R}^d$ that is dominated by some sigma-finite measure $\mu$. Let $p_\theta$ be the density of $P_\theta$ with respect to $\mu$. Let $X$ be a statistical sample drawn from $P_\theta$, and let $Y$ be the output of a channel with transition probability density $p(y|x)$ (with respect to some dominating measure $\nu$ on the sample space $\mathcal{Y}$) when $x$ is the input. In this paper we analyze the Fisher information for estimating $\theta$ from the processed sample $Y$ and show how it can scale with the mutual information $I_\theta(X;Y)$.\footnote{We use $I_\theta(X;Y)$ to denote the mutual information between $X$ and $Y$ when $X$ is drawn from $P_\theta$. In the case that there is a prior distribution on $\theta$ this is the same as $I(X;Y|\theta)$.}

Recall some Fisher information basics. The score for $X$ is defined as
$$S_\theta(X) = \nabla_\theta \log p_\theta(X) \; .$$
The Fisher information matrix for the processed samples $Y$ is
$$I_Y(\theta) = \mathbb{E}\left[\left(\nabla_\theta \log p_\theta(Y)\right)\left(\nabla_\theta \log p_\theta(Y)\right)^T\right]$$
and we can characterize the trace of this matrix by
\begin{equation} \label{eq:decomp}
\mathsf{Tr}(I_Y(\theta)) = \mathbb{E}_Y\|\mathbb{E}_X[S_\theta(X)|Y]\|_2^2 \; .
\end{equation}
The decomposition in \eqref{eq:decomp} is one of the main tools used by the authors in \cite{barnes,ldp_fisher}, and we defer its proof to those references. The proof is a straightforward computation, but requires the interchange of limiting operations, specifically the interchange between integration over the sample space $\mathcal{X}$ and differentiation with respect to the parameter components $\theta_j$. This interchange can be justified under the following regularity conditions (see \cite{borovkov} \S 26 Lemma 1):
\begin{itemize}
\item[(i)] The square-root density $\sqrt{p_\theta(x)}$ is continuously differentiable with respect to each component $\theta_j$ at $\mu$-almost all $x$.
\item[(ii)] The Fisher information for each component $\theta_j$, $\mathbb{E}\left[\left(\frac{\partial}{\partial\theta_j}\log p_\theta(X)\right)^2\right]$, exists and is a continuous function of $\theta_j$.
\end{itemize}
One consequence of these two conditions is that the score random vector has mean zero, i.e.,
$$\mathbb{E}[S_\theta(X)] = 0 \; .$$
One additional regularity condition is needed in the present paper to interchange limits because we are not assuming that the channel $p(y|x)$ has a finite output alphabet $\mathcal{Y}$ (such as in \cite{barnes}), or a point-wise local differential privacy condition (such as in \cite{ldp_fisher}):
\begin{itemize}
\item[(iii)] The channel $p(y|x)$ is square integrable in the sense that $\int p(y|x)^2p_\theta(x)d\mu(x) < \infty$ at $\nu$-almost all $y$ for each $\theta$.
\end{itemize}

Finally, recall some basic facts about sub-Gaussian random variables. We say that a mean-zero random variable $X$ is sub-Gaussian with parameter $N$ if
$$\mathbb{E}\left[e^{\lambda X}\right] \leq e^\frac{\lambda^2N^2}{2}$$
for all $\lambda\in\mathbb{R}$. Furthermore, such random variables $X$ satisfy
$$\mathbb{E}\left[e^\frac{\lambda X^2}{2N^2}\right] \leq \frac{1}{\sqrt{1-\lambda}}$$
for each $\lambda\in[0,1)$ (see, for example, \cite{wainwright} \S 2.1 and \S 2.4).


\section{Fisher Information Bound} \label{sec2}

In this section we state and prove the following main result concerning how Fisher information can scale with the mutual information between $X$ and $Y$.
\begin{thm} \label{thm1}
Suppose that $\langle u,S_\theta(X) \rangle$ is sub-Gaussian with parameter $N$ for any unit vector $u\in\mathbb{R}^d$. Under regularity conditions (i)-(iii) above,
$$\mathsf{Tr}(I_Y(\theta)) \leq 2N^2I_\theta(X;Y) \; .$$
\end{thm}
Note that if $Y\in[1:2^k]$, i.e. $Y$ is constrained to being a $k$-bit message, then $I_\theta(X;Y) \leq H(Y) \leq k$ and we recover the sub-Gaussian case from Theorem 2 of \cite{barnes} as a special case. Similarly, if $p(y|x)$ is a locally differentially private mechanism in the sense that $\frac{p(y|x)}{p(y|x')} \leq e^\varepsilon$ for any $x,x',y$, then $I_\theta(X;Y) = O(\min\{\varepsilon,\varepsilon^2\})$ which recovers Propositions 2 and 4 from \cite{ldp_fisher}, again as a special case.

Note also that the mutual information $I_\theta(X;Y)$ is upper bounded by the capacity of the channel $p(y|x)$ since the capacity is the maximum mutual information over all possible input distributions. In this way we can also think of the above theorem as an upper bound on Fisher information in terms of the capacity of the channel that processes the samples.
\begin{proof}
We begin by ``lifting'' the problem to higher dimensions by considering a new $dB$-dimensional statistical model
$${\bf X} \sim {\bf P}_\theta = \prod_{i=1}^B P_{\theta_i}$$
where ${\bf \theta}=(\theta_1,\ldots,\theta_B)$, ${\bf X}=(X_1,\ldots,X_B)$, and ${\bf Y}=(Y_1,\ldots,Y_B)$. Each $X_i$ is drawn independently according to $P_{\theta_i}$ from the original $d$-dimensional model, and each $Y_i$ is the corresponding output of the channel. Note that
\begin{equation}
\mathsf{Tr}(I_{\bf Y}({\bf \theta})) = \sum_{i=1}^B \mathsf{Tr}(I_{Y_i}(\theta_i))
\end{equation}
and that when $\theta_0 = \theta_1 = \ldots = \theta_B$ we have
\begin{equation}
\mathsf{Tr}(I_{\bf Y}({\bf \theta})) = B\mathsf{Tr}(I_Y(\theta_0)) \; .
\end{equation}
We will therefore proceed by analyzing $\mathsf{Tr}(I_{\bf Y}({\bf \theta}))$ evaluated at the specific ${\bf \theta}$ values with $\theta_0 = \theta_1 = \ldots = \theta_B$. Note that by taking scaled sums of independent sub-Gaussian random variables, the new $dB$-dimensional model has a score function that is sub-Gaussian with the same constant $N$ as that of the original model. We use the decomposition shown in \eqref{eq:decomp} for the new $dB$-dimensional model:
\begin{align}
\mathsf{Tr}(I_{\bf Y}({\bf \theta})) & = \mathbb{E}_{\bf Y}\|\mathbb{E}_{\bf X}[S_{\bf \theta}({\bf X})|{\bf Y}]\|_2^2 \nonumber \\
& =  \mathbb{E}_{\bf Y}\left[\left(\mathbb{E}_{\bf X}\left[\frac{p({\bf Y}|{\bf X})}{p_{\bf \theta}({\bf Y})}\langle u_{\bf Y}, S_\theta({\bf X})\rangle\right]\right)^2\right] \label{eq:ratio}
\end{align}
where $$u_{\bf Y} = \frac{\mathbb{E}_{\bf X}[S_{\bf \theta}({\bf X})|{\bf Y}]}{\|\mathbb{E}_{\bf X}[S_{\bf \theta}({\bf X})|{\bf Y}]\|_2} \; .$$
The key point, and the reason for doing the lifting step, is that the ratio inside the expectation in \eqref{eq:ratio} concentrates around $e^{BI_{\theta_0}(X;Y)}$ as $B$ gets large. More concretely, by the strong law of large numbers,
\begin{align*}
\frac{1}{B}\log \frac{p({\bf y}|{\bf x})}{p_{\bf \theta}({\bf y})} & = \frac{1}{B}\log \prod_{i=1}^B\frac{p(y_i|x_i)}{p_{\theta_0}(y_i)} \\
& = \frac{1}{B}\sum_{i=1}^B \log \frac{p(y_i|x_i)}{p_{\theta_0}(y_i)}
\end{align*}
converges almost surely to
$$\mathbb{E}_{X,Y}\left[\log \frac{p(Y|X)}{p_{\theta_0}(Y)}\right] = I_{\theta_0}(X;Y)$$
as $B \to \infty$. Therefore
$$\frac{p({\bf Y}|{\bf X})}{p_{\bf \theta}({\bf Y})} \leq e^{B(I_{\theta_0}(X;Y)+\varepsilon_B)}$$
with probability at least $1-\varepsilon_B$ for some $\varepsilon_B$ that converges to zero as $B \to \infty$. Let $A_B$ be the event $\left\{({\bf x},{\bf y}) \in \mathcal{X}^B \times \mathcal{Y}^B \; : \; \frac{p({\bf y}|{\bf x})}{p_{\theta_0}({\bf y})} \leq e^{B(I_{\theta_0}(X;Y)+\varepsilon_B)}\right\}$ and let $A_B^C$ be its complement. Following from \eqref{eq:ratio},
\begin{align}
\mathbb{E}_{\bf Y} & \left[\left(\mathbb{E}_{\bf X}\left[\frac{p({\bf Y}|{\bf X})}{p_{\bf \theta}({\bf Y})}\langle u_{\bf Y}, S_\theta({\bf X})\rangle\right]\right)^2\right] \nonumber\\
& \leq \mathbb{E}_{\bf Y}\left[\mathbb{E}_{\bf X}\left[\frac{p({\bf Y}|{\bf X})}{p_{\bf \theta}({\bf Y})}\langle u_{\bf Y}, S_\theta({\bf X})\rangle^2\right]\right] \nonumber\\
 = & \iint_{A_B} \frac{p({\bf y}|{\bf x})}{p_{\bf \theta}({\bf y})}\langle u_{\bf y}, S_\theta({\bf x})\rangle^2 p_\theta({\bf x})p_\theta({\bf y}) \, d\mu({\bf x})d\nu({\bf y}) \label{eq:term1}\\
& + \iint_{A_B^C} \frac{p({\bf y}|{\bf x})}{p_{\bf \theta}({\bf y})}\langle u_{\bf y}, S_\theta({\bf x})\rangle^2 p_\theta({\bf x})p_\theta({\bf y}) \, d\mu({\bf x})d\nu({\bf y}) \label{eq:term2}
\end{align}
We bound the term \eqref{eq:term1} as follows:
\begin{align*}
& \exp \iint 1_{A_B}({\bf x},{\bf y}) \frac{p({\bf y}|{\bf x})}{p_{\bf \theta}({\bf y})}\frac{\lambda}{2}\left(\frac{\langle u_{\bf y}, S_\theta({\bf x})\rangle}{N}\right)^2 \\
&  \quad \quad \quad \quad \quad \quad \quad \quad \quad \quad \quad \quad \cdot p_\theta({\bf x}) p_\theta({\bf y})\, d\mu({\bf x})d\nu({\bf y})  \\
&\leq \iint\frac{p({\bf y}|{\bf x})}{p_{\bf \theta}({\bf y})}\exp\left( 1_{A_B}({\bf x},{\bf y})\frac{\lambda}{2}\left(\frac{\langle u_{\bf y}, S_\theta({\bf x})\rangle}{N}\right)^2\right) \\
&  \quad \quad \quad \quad \quad \quad \quad \quad \quad \quad \quad \quad \cdot p_\theta({\bf x}) p_\theta({\bf y})\, d\mu({\bf x})d\nu({\bf y})  \\
&\leq \varepsilon_B + \iint_{A_B}\frac{p({\bf y}|{\bf x})}{p_{\bf \theta}({\bf y})}\exp\left(\frac{\lambda}{2}\left(\frac{\langle u_{\bf y}, S_\theta({\bf x})\rangle}{N}\right)^2\right) \\
&  \quad \quad \quad \quad \quad \quad \quad \quad \quad \quad \quad \quad \cdot p_\theta({\bf x}) p_\theta({\bf y})\, d\mu({\bf x})d\nu({\bf y})  \\
& \leq \varepsilon_B + e^{B(I_{\theta_0}(X;Y)+\varepsilon_B)} \iint_{A_B}\exp\left(\frac{\lambda}{2}\left(\frac{\langle u_{\bf y}, S_\theta({\bf x})\rangle}{N}\right)^2\right) \\
&  \quad \quad \quad \quad \quad \quad \quad \quad \quad \quad \quad \quad \cdot p_\theta({\bf x}) p_\theta({\bf y})\, d\mu({\bf x})d\nu({\bf y})  \\
& \leq \varepsilon_B + \frac{e^{B(I_{\theta_0}(X;Y)+\varepsilon_B)}}{\sqrt{1-\lambda}} \; .
\end{align*}
for $0\leq\lambda<1$. Taking logs,
\begin{align} \label{eq:term3}
\iint_{A_B} & \frac{p({\bf y}|{\bf x})}{p_{\bf \theta}({\bf y})}\langle u_{\bf y}, S_\theta({\bf x})\rangle^2 p_\theta({\bf x})  p_\theta({\bf y})\, d\mu({\bf x})d\nu({\bf y}) \nonumber \\
& \leq \frac{2N^2}{\lambda}\log\left(\varepsilon_B + \frac{e^{B(I_{\theta_0}(X;Y)+\varepsilon_B)}}{\sqrt{1-\lambda}}\right) \; .
\end{align}
For the other term \eqref{eq:term2},
\begin{align}
& \iint_{A_B^C} \frac{p({\bf y}|{\bf x})}{p_{\bf \theta}({\bf y})}\langle u_{\bf y}, S_\theta({\bf x})\rangle^2 p_\theta({\bf x})p_\theta({\bf y}) \, d\mu({\bf x})d\nu({\bf y}) \nonumber\\
& \leq \iint_{A_B^C} p_\theta({\bf x},{\bf y})\|S_\theta({\bf x})\|_2^2 \, d\mu({\bf x})d\nu({\bf y}) \nonumber \\
& \leq \left( \varepsilon_B \iint p_\theta({\bf x},{\bf y})\|S_\theta({\bf x})\|_2^4d\mu({\bf x})d\nu({\bf y}) \right)^\frac{1}{2} \label{eq:term5}\\
& = \left(\varepsilon_B \int p_\theta({\bf x})\|S_\theta({\bf x})\|_2^4d\mu({\bf x})\right)^\frac{1}{2} \leq \left(c_1\varepsilon_B(dB)^2N^4\right)^\frac{1}{2} \label{eq:term4}
\end{align}
for some absolute constant $c_1$. To get \eqref{eq:term5} we have used the Cauchy-Schwarz inequality, and \eqref{eq:term4} follows by using the sub-Gaussianity of each of the $(dB)^2$ different terms in $\|S_\theta(X)\|_2^4$ to bound their fourth moments. Combining \eqref{eq:term3} and \eqref{eq:term4},
\begin{align*}
\mathsf{Tr}(I_{Y}(\theta_0)) & = \frac{1}{B}\mathsf{Tr}(I_{\bf Y}({\bf \theta})) \\
& \leq \frac{2N^2}{\lambda B}\log\left(\varepsilon_B + \frac{e^{B(I_{\theta_0}(X;Y)+\varepsilon_B)}}{\sqrt{1-\lambda}}\right) \\
& \quad + (c_1\varepsilon_Bd^2N^4)^\frac{1}{2}
\end{align*}
and taking $B\to\infty$ and then $\lambda \to 1$ gives the final result.
\end{proof}

\section{Distributed Statistical Estimation}
The upper bound on Fisher information from Section \ref{sec2} can be of particular interest in a distributed setting, where there are $n$ distinct nodes each with a sample $X_i$ taken i.i.d. from $P_\theta$. In this setting, the nodes communicate information about their samples via a ``public blackboard'' in multiple rounds of communication, and then a centralized estimator uses the blackboard transcript after communication to construct an estimate $\hat\theta$ of the parameter $\theta$.

More formally, on round $t=1,\ldots,T$ of communication, each node $i=1,\ldots,n$ communicates a random variable $Y_{i,t}$ according to its local data $X_i$ and the previous information written on the transcript $Y_{i-1,t},Y_{i-2,t},\ldots,Y_{1,t},\Pi_{t-1},\ldots,\Pi_1$ via a channel $p(y_{i,t}|x_i,y_{i-1,t},\ldots,y_{1,t},\pi_{t-1},\ldots,\pi_1)$. Here we define $\pi_t=(y_{1,t},\ldots,y_{n,t})$ to be the information written to the public blackboard after communication on round $t$. The estimator $\hat\theta$ is then a function of the total blackboard transcript
$$\Pi = (\Pi_1,\ldots,\Pi_T) \; . $$
By using Theorem \ref{thm1}, we have the following upper bound on Fisher information from the transcript $\Pi$.
\begin{cor} \label{cor1}
Suppose that $\langle u,S_\theta(X) \rangle$ is sub-Gaussian with parameter $N$ for any unit vector $u\in\mathbb{R}^d$. Under regularity conditions (i)-(iii) above,
$$\mathsf{Tr}(I_\Pi(\theta)) \leq 2N^2I_\theta(X_1,\ldots,X_n;\Pi) \; .$$
\end{cor}
\begin{proof}
Using the equivalent of \eqref{eq:decomp} in this interactive communication setting,
\begin{align}
\mathsf{Tr}(I_\Pi(\theta)) = \sum_{i=1}^n \mathbb{E}_\Pi \left\|\frac{\mathbb{E}_{X_i}\left[S_{\theta}(X_i)p_{i,\Pi}(X_i)\right]}{\mathbb{E}_{X_i}\left[p_{i,\Pi}(X_i)\right]}\right\|^2 \label{eq:bb_term1}
\end{align}
where $$p_{i,\pi}(x_i) = \prod_{t} p(y_{i,t}|x_i,y_{i-1,t},\ldots,y_{1,t},\pi_{t-1},\ldots,\pi_1) \; .$$
The decomposition \eqref{eq:bb_term1} is detailed in Appendix E of \cite{ldp_fisher}. Using Theorem \ref{thm1} above, each term in the sum in \eqref{eq:bb_term1} can be upper bounded by $2N^2I_\theta(X_i;\Pi)$. Thus by the independence of the $X_i$,
$$\mathsf{Tr}(I_\Pi(\theta)) \leq 2N^2I_\theta(X_1,\ldots,X_n;\Pi) \; .$$
\end{proof}

We now consider the special case of the Gaussian location model where $P_\theta = \mathcal{N}(\theta,\sigma^2I_d)$ and $\Theta = [-1,1]^d$. It can be readily checked that in this case
$S_\theta(X) \sim \mathcal{N}\left(0,\frac{1}{\sigma^2}I_d\right).$ This leads to the following upper bound on Fisher information, where we will see shortly that the constant on the right-hand side is optimal.
\begin{cor} \label{cor2}
Suppose that $P_\theta = \mathcal{N}(\theta,\sigma^2I_d)$. Then
$$\mathsf{Tr}(I_\Pi(\theta)) \leq \frac{2}{\sigma^2}I_\theta(X_1,\ldots,X_n;\Pi) \; .$$
\end{cor}
By using Corollary \ref{cor2} along with the multivariate van Trees inequality from \cite{gill}, we immediately get the following lower bound on the minimax risk in estimating $\theta$ from the transcript $\Pi$.
\begin{cor} \label{cor3}
Suppose that $P_\theta = \mathcal{N}(\theta,\sigma^2I_d)$ and $\Theta = [-1,1]^d$. Then
$$\sup_{\theta\in\Theta} \mathbb{E}\|\hat\theta(\Pi)-\theta\|_2^2 \geq \frac{d^2}{\frac{2}{\sigma^2}\sup_{\theta\in\Theta}I_\theta(X_1,\ldots,X_n;\Pi)+\pi^2d} \; .$$
\end{cor}
In order to see that the constant 2 in Corollaries \ref{cor2} and \ref{cor3} is optimal, consider the following example. Let $d=1$ and $T=1$ and suppose that communication is done independently over additive white Guassian noise channels so that $Y_i = X_i + W_i$ where $W_i\sim\mathcal{N}(0,\sigma_\text{noise}^2)$ and $\Pi=(Y_1,\ldots,Y_n)$. In this case 
$$I_\theta(X_1,\ldots,X_n;\Pi) = \frac{n}{2}\log\left(1+\frac{\sigma^2}{\sigma_\text{noise}^2}\right)$$
and Corollary \ref{cor3} gives a lower bound of
$$\sup_{\theta\in\Theta} \mathbb{E}[(\hat\theta(\Pi)-\theta)^2] \geq \frac{\sigma^2}{n\log\left(1+\frac{\sigma^2}{\sigma_\text{noise}^2}\right)}$$
assuming that $n$ is large enough so that the second term in denominator is negligible. The simple averaging estimator $\hat\theta(\Pi) = \frac{1}{n}\sum_{i=1}^n Y_i$ is unbiased and has variance $\frac{\sigma^2}{n}\left(1+\frac{\sigma_\text{noise}^2}{\sigma^2}\right).$ In the regime where $\sigma_\text{noise}^2 >> \sigma^2$, both the lower bound and expected squared error become $\frac{\sigma_\text{noise}^2}{n}$, and thus any constant less than 2 in Corollary \ref{cor2} would lead to a contradiction.

The distributed estimation setting considered here is the same setting that is considered in past works \cite{barnes,ldp_fisher}, except that here we do not assume the channels $p(y_{i,t}|x_i,y_{i-1,t},\ldots,y_{1,t},\pi_{t-1},\ldots,\pi_1)$ have a particular communication or privacy-constrained structure, and instead leave them general and have the constraint be the total mutual information $I_\theta(X_1,\ldots,X_n;\Pi)$. The distributed data-processing inequality from \cite{braverman} can also apply to general mutual information constrained settings, but requires a bounded likelihood ratio assumption that we do not make. One of the benefits of our Fisher information bound is that it can be applied to the Gaussian location model directly, without the need to truncate the Gaussians so that they satisfy this bounded likelihood ratio assumption. In the subsequent section we show how our Fisher information bounds can also imply Jensen-Shannon divergence based strong data processing inequalities that are very similar to those from \cite{braverman}.

\section{Relation to Divergence Bounds}
Because of Fisher information's interpretation as a second-order approximation of KL divergence (or Jensen-Shannon divergence) as one moves around the parameter space $\Theta$, the above upper bounds on Fisher information can also imply similar upper bounds on the divergence between two distributions from $\{Q_\theta\}_{\theta\in\Theta}$ where $Q_\theta = P_\theta^n \circ P_{\Pi|X_1,\ldots,X_n}$ is the induced distribution for $\Pi$. 

Instead of working with KL divergence directly, we will instead analyze the related Jensen-Shannon divergence because of its nicer properties. In particular, its square root satisfies the triangle inequality which we describe and use below. Let
$$\mathsf{JS}(P\|Q) = \mathsf{KL}\left(P\bigg\|\frac{P+Q}{2}\right) + \mathsf{KL}\left(Q\bigg\|\frac{P+Q}{2}\right)$$
be the Jensen-Shannon divergence between distributions $P$ and $Q$ where $\mathsf{KL}(P\|Q)$ is the usual KL divergence. The square-root of the Jensen-Shannon divergence satisfies the triangle inequality \cite{endres} in that
$$\sqrt{\mathsf{JS}(P\|Q)} \leq \sqrt{\mathsf{JS}(P\|R)} + \sqrt{\mathsf{JS}(Q\|R)}$$
for any distributions $P,Q,$ and $R$. 

In this section we will need the following additional regularity condition.
\begin{itemize}
\item[(iv)] Suppose that $\mathsf{JS}(Q_\theta\|Q_{\theta+\Delta\theta})$ can be represented by its Taylor expansion
 $$\mathsf{JS}(Q_\theta\|Q_{\theta+\Delta\theta}) = \frac{1}{4}\Delta\theta^T I_{\Pi}(\theta) \Delta\theta + O(\|\Delta\theta\|^3)$$
in such a way that the constants in the Big-O term can be made independent of the choice of $\theta\in\Theta$.
\end{itemize}For the Jensen-Shannon divergence we will prove the following data processing bound.
\begin{thm} \label{thm2}
Suppose $\theta_0,\theta_1\in\Theta$ are such that $\theta_\lambda = \lambda \theta_1 + (1-\lambda)\theta_0$ for $\lambda\in[0,1]$ are contained in $\Theta$. Under the assumptions in Theorem \ref{thm1} above and regularity condition (iv),
$$\mathsf{JS}(Q_{\theta_0}\|Q_{\theta_1}) \leq \frac{\|\theta_1-\theta_0\|_2^2N^2}{2}\int_0^1 I_{\theta_\lambda}(X_1,\ldots,X_n;\Pi)d\lambda \; .$$
\end{thm}
Note that if we consider the random variable $V$ to represent a prior that chooses the parameter $\theta_0$ or $\theta_1$ each with probability $\frac{1}{2}$, then
$$I(V;\Pi) = \frac{1}{2}\mathsf{JS}(Q_{\theta_0}\|Q_{\theta_1})$$
and we can write the result from Theorem \ref{thm2} as
\begin{equation}I(V;\Pi) \leq \frac{\|\theta_1-\theta_0\|_2^2N^2}{4}\int_0^1 I_{\theta_\lambda}(X_1,\ldots,X_n;\Pi)d\lambda \; . \label{eq:thm2} \end{equation}
We compare this to the following theorem. Let $\beta(P_{\theta_0},P_{\theta_1})$ be the strong data processing inequality (SDPI) constant defined to be the minimum value $\beta$ such that
$$I(V;\Pi) \leq \beta I(X;\Pi)$$
where $V\to X\to\Pi$ is a Markov chain and $V\sim\text{Bern}(1/2)$ picks whether $X$ is drawn from $P_{\theta_0}$ or $P_{\theta_1}$ as above.
\begin{thm}[\cite{braverman} Theorem 1.1] \label{thm3}
Suppose $P_{\theta_0},P_{\theta_1}$ are two probability distributions with $\frac{1}{c}P_{\theta_0}\leq P_{\theta_1}\leq c P_{\theta_0}$ for some constant $c\geq 1$. Let $\beta(P_{\theta_0},P_{\theta_1})$ be the SDPI constant defined above. Then
\begin{equation}I(V;\Pi) \leq Kc\beta(P_{\theta_0},P_{\theta_1})\min_{v\in\{\theta_0,\theta_1\}}I_{v}(X_1,\ldots,X_n;\Pi) \label{eq:thm3} \end{equation}
for a universal constant $K$.
\end{thm}
Note that in Theorem \ref{thm3} we have only presented the special case $V\sim\text{Bern}(1/2)$ so that the left-hand side can be written as the mutual information, whereas the general theorem does not assume even probabilities and instead the left-hand side is written as a Hellinger distance term.

To compare \eqref{eq:thm2} and \eqref{eq:thm3} consider what happens in the Gaussian case where $P_\theta\sim\mathcal{N}(\theta,\sigma I_d)$. In this case the SDPI constant $\beta(P_{\theta_0},P_{\theta_1}) = O\left( \frac{\|\theta_0-\theta_1\|_2^2}{\sigma^2}\right)$ as used in \cite{braverman} and shown in \cite{raginsky}. In this case the right-hand sides of the two bounds are very similar with only some minor differences. In particular, \eqref{eq:thm3} has the minimum mutual information over the two parameter values, but has an extra factor of $c$ to compensate. In contrast, \eqref{eq:thm2} is written in terms on the average mutual information along the linear path from $\theta_0$ to $\theta_1$ in the parameter space. One benefit of \eqref{eq:thm2} over \eqref{eq:thm3} is that it does not need this bounded likelihood ratio assumption, instead assuming the sub-Gaussian score function property from Theorem \ref{thm1}, and therefore it can be applied to the Gaussian case directly without any need for truncation.

\subsection{Proof of Theorem \ref{thm2}}
Using a multivariate Taylor expansion of $\mathsf{JS}(Q_\theta\|Q_{\theta+\Delta\theta})$  at $\Delta\theta=0$,
\begin{align} 
\mathsf{JS}(Q_\theta\|Q_{\theta+\Delta\theta}) = & \frac{1}{4}\Delta\theta^TI_\Pi(\theta)\Delta\theta + O(\|\Delta\theta\|^3) \; . \label{eq:taylor}
\end{align}
Given two parameter values $\theta_0,\theta_1\in\Theta\subseteq\mathbb{R}^d$ we can upper bound the Jensen-Shannon divergence between $Q_{\theta_0}$ and $Q_{\theta_1}$ as follows. For any real number $0\leq\lambda\leq 1$ we define $\theta_\lambda = \lambda \theta_1 + (1-\lambda)\theta_0$. By the triangle inequality,
\begin{align*} 
\sqrt{\mathsf{JS}(Q_{\theta_0}\|Q_{\theta_1})} & \leq \sum_{i=1}^M \sqrt{\mathsf{JS}\left(Q_{\theta_\frac{i-1}{M}}\bigg\|Q_{\theta_\frac{i}{M}}\right)}
\end{align*}
and then by squaring and using Jensen's inequality,
\begin{align} 
& \mathsf{JS}(Q_{\theta_0}\|Q_{\theta_1}) \nonumber \\
& \leq \left(\sum_{i=1}^M \sqrt{\mathsf{JS}\left(Q_{\theta_\frac{i-1}{M}}\bigg\|Q_{\theta_\frac{i}{M}}\right)}\right)^2 \nonumber \\
& \leq M\sum_{i=1}^M \mathsf{JS}\left(Q_{\theta_\frac{i-1}{M}}\bigg\|Q_{\theta_\frac{i}{M}}\right) \nonumber \\
& = \frac{M}{4}\sum_{i=1}^M\left(\frac{\theta_1-\theta_0}{M}\right)^TI_\Pi\left(\theta_\frac{i-1}{M}\right)\left(\frac{\theta_1-\theta_0}{M}\right) \nonumber \\
& \quad + MO\left(\left\|\frac{\theta_1-\theta_0}{M}\right\|^3\right) \label{eq:int}
\end{align}
where the last line \eqref{eq:int} follows from \eqref{eq:taylor}. Continuing from \eqref{eq:int},
\begin{align*}
 \mathsf{JS}(Q_{\theta_0}\|Q_{\theta_1}) \leq & \frac{1}{4}\|\theta_1-\theta_0\|_2^2\frac{1}{M}\sum_{i=1}^M \mathsf{Tr}\left( I_\Pi\left(\theta_\frac{i-1}{M}\right)\right) \\
& \quad + MO\left(\left\|\frac{\theta_1-\theta_0}{M}\right\|^3\right) \; .
\end{align*}
Taking $M\to\infty$ we get that
\begin{align*}
\mathsf{JS}(Q_{\theta_0}\|Q_{\theta_1}) \leq  \frac{1}{4}\|\theta_1-\theta_0\|_2^2 \int_0^1\mathsf{Tr}\left( I_\Pi\left(\theta_\lambda\right)\right)d\lambda \; .
\end{align*}
where we have used the Riemann integrability of the entries of $I_\Pi(\theta_\lambda)$ which can be shown using regularity conditions (i) and (ii). The result then follows from Corollary \ref{cor1}.
\subsection{One-Parameter Families of Distributions}
In this section we apply Theorem \ref{thm2} to the case where we are given two distributions $\mu_0$ and $\mu_1$ that do satisfy the bounded likelihood ratio assumption, even if they are not part of a parametric family. Suppose the two distributions have densities $f_0(x)$ and $f_1(x)$, respectively. In this case we can define a one-parameter family of distributions between the two using an exponential twist such as in \cite{relations}. For $\theta\in[0,1]$ we define a new density $f_\theta$ by
$$f_\theta(x) = \frac{1}{C_\theta}f_1^\theta(x)f_0^{1-\theta}(x)$$
where $C_\theta = \int f_1^\theta(x)f_0^{1-\theta}(x)dx$ in order to normalize the density. The score function for this one-parameter family is
$$S_\theta(x) = \log\frac{f_1(x)}{f_0(x)} - \frac{C'_\theta}{C_\theta} \; .$$
If, like in \cite{braverman}, we make the bounded likelihood ratio assumption $$\frac{1}{c}f_0(x) \leq f_1(x) \leq cf_0(x) \; ,$$
then
$|S_\theta(x)| \leq 2\log c$ for all $x$. This implies the score function is sub-Gaussian with a parameter that is $O(\log c)$, and yields
$$I(V;\Pi) \leq K(\log c)^2 \int_0^1 I_{\lambda}(X_1,\ldots,X_n;\Pi)d\lambda$$
for an absolute constant $K$.

\bibliographystyle{IEEEtran}
\bibliography{main.bib}


\end{document}